\documentclass[12pt,fulpage]{article}

\usepackage[pdftex]{hyperref}
\usepackage{fullpage,amssymb,amsmath,amsthm}
\usepackage{latexsym}
\usepackage{verbatim}

\newcommand{\remove}[1]{}

\newtheorem{theorem}{Theorem}[section]
\newtheorem{thm}{Theorem}[section]

\newtheorem{lemma}[thm]{Lemma}

\newtheorem{definition}[thm]{Definition}

\newtheorem{corollary}[thm]{Corollary}

\def\F{{\mathbb{F}}}
\def\R{{\mathbb{R}}}

\def\poly{{\mathrm{poly}}}

\def\B{{\{0,1\}}}

\def\Im{{\text{Image}}}
\def\wt{{\text{weight}}}
\begin{document}

\title{Capacity achieving multiwrite WOM codes}

\author{Amir Shpilka\thanks{Faculty of Computer Science, Technion --- Israel
Institute of Technology, Haifa, Israel, {\tt
shpilka@cs.technion.ac.il}. This research was partially supported
by the Israel Science Foundation (grant number 339/10).}}

%\author{Amir Shpilka\thanks{This research was partially supported
%by the Israel Science Foundation (grant number 339/10).}}
%\institute{ Faculty of Computer Science\\ Technion--- Israel
%Institute of Technology\\ Haifa 32000, Israel\\
%\email{ shpilka@cs.technion.ac.il} }

\date{}

\maketitle

\begin{abstract}

In this paper we give an explicit construction of a capacity achieving family of binary $t$-write WOM codes for any number of writes $t$, that have a polynomial time encoding and decoding algorithms. The block length of our construction is $N=(t/\epsilon)^{O(t/(\delta\epsilon))}$ when $\epsilon$ is the gap to capacity and encoding and decoding run in time $N^{1+\delta}$.  This is the first deterministic construction achieving these parameters. 
Our techniques also apply to larger alphabets. 

\end{abstract}
%\thispagestyle{empty}
%\newpage
%\pagenumbering{arabic}

\section{Introduction}

In \cite{RivestShamir82} Rivest and Shamir introduced the notion
of {\em write-once-memory} and showed its relevance to the problem
of saving data on optical disks. A write-once-memory, over the
binary alphabet, allows us to change the value of a memory cell
(say from $0$ to $1$) only once. Thus, if we wish to use the
storage device for storing $t$ messages in $t$ rounds, then we
need to come up with an encoding scheme that allows for $t$-write
such that each memory cell is written at most one time in the binary setting, or it is written monotonically increasing values in the $q$-ary setting. An
encoding scheme satisfying these properties is called a $t$-write
Write-Once-Memory code, or a $t$-write WOM code for short. This model has
recently gained renewed attention due to similar problems that
arise when using flash memory devices. We refer the readers to
\cite{YKSVW10,GYDSVW11} for a more detailed introduction to WOM codes and
their use in encoding schemes for flash memory.

One interesting goal concerning WOM codes is to find codes that
have good rate for $t$-write. Namely, to find encoding schemes
that allow to save the maximal information-theoretic amount of
data possible under the write-once restriction. Following
\cite{RivestShamir82} it was shown that the capacity (i.e. maximal
rate) of $t$-write  binary WOM code is\footnote{All logarithms in
this paper are taken base $2$.} $\log(t+1)$ (see
\cite{RivestShamir82,Heegard84,FuVinck99}). Stated differently, if
we wish to use an $n$-bit memory $t$-times then each time we can
store, on average, $n\cdot \log(t+1)/t$ many bits.

\remove{
Currently, the
best known explicit encoding scheme for two-write  (over the
binary alphabet) has rate roughly $1.49$ (compared to the optimal
$\log 3 \approx 1.585$) \cite{YKSVW10}. We note that these codes,
of rate $1.49$, were found using the help of a computer search. A
more `explicit' construction given in  \cite{YKSVW10} achieves
rate $1.46$.
}

In this work we give, for any number of writes $t$, an explicit construction of an encoding scheme that achieves the theoretical capacity. Before describing our results we give
a formal definition of a $t$-write WOM code.

For two vectors of the same length $y$ and $y'$ we say that
$y'\leq y$ if $y'_i\leq y_i$ for every coordinate $i$.

%\section{Definitions}
%A WOM consists of binary memory elements that can only be changed
%from a zero state to a one state. We address the problem of
%designing WOM-codes that achieve the theoretical capacity for the
%case of two rounds of writing to the memory cells.

\begin{definition}
A $t$-write binary WOM of length $n$ over the sets of messages
$\Omega_1, \Omega_2,\ldots,\Omega_t$  consists of $t$ encoding functions $E_1,\ldots,E_t$ such that
$E_1:\Omega_1 \to \B^n$, $E_2: E_1(\Omega_1)\times \Omega_2 \to
\B^n$ and in general, $E_i:\Im(E_{i-1})\times \Omega_i\to \B^n$, and $t$ decoding functions $D_i:\Im(E_i) \to \Omega_i$ that
satisfy the following properties.\footnote{To simplify notation let us denote by $E_0$ a map whose image is $\Im(E_0)=\vec{0}$. We also abuse notation and denote $E_1(\vec{0},x)=E_1(x)$. }
\begin{enumerate}
  \item For all $1\leq i\leq t$,  $c \in \Im(E_{i-1})$ and  $x_i\in \Omega_i$, $D_i(E_i(c,x_i))=x_i$.
  \item For all $1\leq i\leq t$,  $c \in \Im(E_{i-1})$ and  $x_i\in \Omega_i$ we have
  that $c\leq E_i(c,x_i)$.
\end{enumerate}
The rate of such a WOM code is defined to be $(\log|\Omega_1| +
\log|\Omega_2|+\ldots+\log|\Omega_t|)/n$.
\end{definition}

Intuitively, the definition enables the encoder to use $E_1$ as
the encoding function in the first round. If the message $x_1$ was
encoded (as the string $E_1(x_1)$) and then we wished to encode in
the second round the message $x_2$, then we write the string
$E_2(E_1(x_1),x_2)$, etc. Since $E_1(x_1)\leq E_2(E_1(x_1),x_2)$, we
only have to change a few zeros to ones (in the binary setting) in order to move from
$E_1(x_1)$ to $E_2(E_1(x_1),x_2)$. Similarly for any one of the $t$ writes. The requirement on the decoding
functions guarantees that at each round we can
correctly decode the memory.\footnote{We implicitly assume that
the decoder knows, given a codeword, the round in which it was encoded. At worst this can add another $t$-bits to
the encoding which has no real effect (in the asymptotic sense) on the
rate.} Notice, for example, that in the $i$th round we are only required to
decode $x_i$ and not any of the earlier messages.

Similarly, one can also define WOM codes over other
alphabets. 
%but in this paper we will  be mostly interested in the binary alphabet.

The capacity region of a binary $t$-write binary WOM code was shown in \cite{RivestShamir82,FuVinck99,Heegard84} to be the set 
\begin{align*}
C_t = \left\{ (R_1,\ldots,R_t)\in\R_{+}^t \mid  R_1 \leq H(p_1),  R_2\leq (1-p_1)H(p_2), R_{3} \leq \prod_{i=1}^{2}(1-p_i) H(p_{3}),\right. \quad\quad \\  \left. \ldots, R_{t-1} \leq \prod_{i=1}^{t-2}(1-p_i) H(p_{t-1}), R_t\leq \prod_{i=1}^{t-1}(1-p_i), \text{ where } 0\leq p_1,\ldots,p_{t-1}\leq 1/2 \right\}.
\end{align*}
From this it is not hard to show that the maximal rate of a $t$-write WOM code is $\log(t+1)$ and that this rate is achieved by the point 
\begin{equation*}
\left(R_1,\ldots,R_t\right)=\left(H\left(\frac{1}{t+1}\right),\frac{t}{t+1}H\left(\frac{1}{t}\right),\frac{t-1}{t+1}H\left(\frac{1}{t-1}\right),\ldots,\frac{3}{t+1}H\left(\frac{1}{3}\right),\frac{2}{t+1}\right)
\end{equation*} 
in the capacity region (corresponding to the choice $(p_1,\ldots,p_{t-1})=(\frac{1}{t+1},\frac{1}{t},\ldots,\frac{1}{3})$).
For a point $(R_1,\ldots,R_t)\in C_t$ we call the vector $(p_1,\ldots,p_{t-1})\in [0,\frac{1}{2}]^{t-1}$ the corresponding weight vector. Namely, $(p_1,\ldots,p_{t-1})$ is a corresponding weight vector if $R_1 \leq H(p_1), R_2\leq (1-p_1)H(p_2),\ldots, R_{t-1} \leq \prod_{i=1}^{t-2}(1-p_i) H(p_{t-1})$ and $R_t\leq \prod_{i=1}^{t-1}(1-p_i)$.

\subsection{Our results}\label{sec: our results}

We provide an encoding scheme that for any $\epsilon>0$ and  $(R_1,\ldots,R_t)$ in the capacity region achieves rate $R_1+\ldots+R_t-\epsilon$ with polynomial time encoding and decoding schemes. 

\begin{theorem}\label{thm:main}
  For any $\epsilon,\delta>0$ and $(R_1,\ldots,R_t)\in C_t$  there exists an integer $N=N(\epsilon,\delta, t)$ and an explicit construction of a $t$-write binary WOM code of length $N$ of rate at least $\sum_{i=1}^{t}R_i -\epsilon$.
  Furthermore, the encoding and decoding functions run in time $ N^{1+\delta}$.
\end{theorem}

In particular, for $(R_1,\ldots,R_t)=(H(\frac{1}{t+1}),\frac{t}{t+1}H(\frac{1}{t}),\frac{t-1}{t+1}H(\frac{1}{t-1}),\ldots,\frac{3}{t+1}H(\frac{1}{3}),\frac{2}{t+1})$ we give a construction of a binary WOM code of rate $\log(t+1)-\epsilon$. If we wish to
achieve a polynomial time encoding and decoding then our proof
gives the bound $N(\epsilon,\delta,t)=
(\frac{t}{\epsilon})^{O(\frac{t}{\delta\epsilon})}$. If we wish to have a short block
length, e.g. $N=\poly(1/\epsilon)$, then our running time
deteriorates and becomes $N^{O(t/\epsilon)}$.

A completely analogous construction also works for arbitrary alphabets. However, here it is more difficult to explicitly characterize the capacity region. Hence, our theorem for this case basically says that given realizable parameters we can construct a scheme achieving them (up to $\epsilon$ loss in rate). We shall use the following notation. Let $S_1,\ldots,S_t$ be lower-triangular column-stochastic matrices of dimension $q\times q$. We shall think of the $(j_2,j_1)$ entry of $S_i$ (when $j_1,j_2\in\{0,1,2,\ldots,q-1\}$) as describing the fraction of memory cells that had value $j_1$ before the $i$th write and then were written the value $j_2$ in the $i$th round, where the fraction is computed relative to the number of memory cells holding the value $j_1$. Thus, these matrices describe the distribution of values that we expect to see after each write and also how values change in the process. For example, in $S_1$ only the first column is meaningful as the memory initially holds the value $0$. We say that $S_1,\ldots,S_t$ is in the capacity region of $t$-write WOM codes over $\{0,\ldots,q-1\}$ if there exists an encoding scheme realizing $S_1,\ldots,S_t$ (i.e. a scheme that abides the constraints described by the $S_j$s). 
%We say that an encoding scheme $\epsilon$-realizes $S_1,\ldots,S_t$, if it realizes the lower diagonal stochastic %matrices $S'_1,\ldots,S'_t$, such that for any $i,j_1,j_2$, $|(S_i)_{j_2,j_1}-(S_i)_{j_2,j_1}|\leq \epsilon$. 

\begin{theorem}\label{thm:main-large}
  For any $\epsilon,\delta>0$ and $(S_1,\ldots,S_t)$ in the capacity region of $t$-write WOM over the alphabet $\{0,\ldots,q-1\}$ there exists an integer $N_q=N_q(\epsilon,\delta, t)$ and an explicit construction of a $t$-write $q$-ary WOM code of length $N_q$ that 
  %$\epsilon$-realizes $S_1,\ldots,S_t$ and that 
  has rate which is $\epsilon$-close to the rate of the scheme realizing $S_1,\ldots,S_t$.
  Furthermore, the encoding and decoding functions run in time $ N_q^{1+\delta}$.
\end{theorem}

As before we have the bound $N_q(\epsilon,\delta,t)=
(\frac{t}{\epsilon})^{O(\frac{t}{\delta\epsilon})}$, where now the big Oh may depend on $q$ as well. Since our proof for the binary setting contains all the relevant ideas, and the extension to other alphabets is (more or less) obvious, we shall omit most of the details of the non-binary case.

\subsection{Comparison to the brute force scheme}\label{sec:brute}

The first observation that one makes is that the problem of
approaching capacity is, in some sense, trivial. This basically
follows from the fact that concatenating WOM codes (in the sense
of string concatenation) does not hurt any of their properties.
Thus, if we can find, even in a brute force manner, a code of
length $m$ that is $\epsilon$-close to capacity, in time $T(m)$,
then concatenating $N=T(m)$ copies of this code, gives a code of
length $Nm$ whose encoding algorithm takes $NT(m)=N^2$ time.
Notice however, that for the brute force algorithm, $T(m) \approx
2^{2^m}$ and so, to get $\epsilon$-close to capacity we need
$m\approx 1/\epsilon$ and thus $N \approx 2^{2^{1/\epsilon}}$.

\remove{
The same argument also shows that finding capacity approaching WOM
codes for $t$-write, for any constant $t$, is ``easy'' to achieve
in the asymptotic sense, with a polynomial time encoding/decoding
functions, given that one is willing to let the encoding length
$n$ be obscenely huge.
}

In fact, following Rivest and Shamir,
Heegard actually showed that a randomized encoding scheme can
achieve capacity for all $t$ \cite{Heegard84}, but, naturally, such encoding schemes do not have efficient decoding algorithms.

In view of that, our construction can be seen as giving a big
improvement over the brute force construction. Indeed, we only
require $N \approx {(t/\epsilon)^{t/\epsilon}}$ and we give encoding and
decoding schemes that can be implemented efficiently (as functions of the block length).

We later discuss a connection with {\em invertible extractors
for hamming balls} and show that an efficient construction of such objects could lead to capacity-achieving WOM codes of
reasonable block length.

\subsection{Comparison to earlier works}\label{sec:other works}

As we explain in Section~\ref{sec:method} our work is close in spirit to \cite{Shpilka-WOM12}, but whereas the latter dealt mainly with the case of $2$-write WOM codes, here we develop a scheme for any number of writes. 

\medskip

Prior to our work the best {\em deterministic} encoding schemes for several writes where obtained in the works \cite{KYSVW10,YaakobiShpilka12}. As we achieve capacity here and previous constructions where at some distance from capacity, which grew with the number of writes, we do not give a detailed comparison of the different results and just mentioned that except the work \cite{Shpilka-WOM12} no other work was guaranteed to achieve capacity. (Although we soon discuss a different work that does achieve capacity, but has some failure probability)
On the other hand, while \cite{KYSVW10,YaakobiShpilka12} did not achieve capacity they did give some multiple write codes of a significantly shorter block length than what we achieve here. So, in that respect, they give codes that may be more practical than our new construction. 

\medskip 

In \cite{BurshteinS12} Burshtein and Strugatski gave a construction of capacity achieving binary WOM codes based on polar codes. Their construction is based on the work of \cite{KoradaU10} that uses polar codes for source coding. A small downside of their construction is that it is guaranteed to work with high probability, but there is always some small chance of failure. In particular, in contrast to our construction, it is not always guaranteed to succeed. On the other hand, since \cite{BurshteinS12} rely on polar codes that have very efficient encoding and decoding schemes, it may be the case that they can make their construction work with a block length that is polynomial in $1/\epsilon$, which is much better than our block length which is exponentially large in $1/\epsilon$. We note however, that it is not clear that their block length is indeed only polynomial in $1/\epsilon$, as it is a bit difficult to analyze the dependence between distortion-rate (which affects the gap to capacity in the ensuing construction of \cite{BurshteinS12}) and block-length in the construction of \cite{KoradaU10}. It may very well be the case that although polar codes are so efficient, we still encounter a significant blow up in block-length when pushing both distortion-rate and failure probability parameters down. Nevertheless, as we said, it may still be the case that they can get codes of $\poly(1/\epsilon,\log(1/\delta))$ block-length which are $\epsilon$-close to capacity and have failure probability at most $\delta$, in which case these codes will have more reasonable parameters than what we achieve here. 

%\subsection{Organization}

\subsection{Notation}
For a $k\times m$ matrix $A$ and a subset $S\subset [m]$ we let
$A|_S$ be the $k\times |S|$ sub-matrix of $A$ that contains only
the columns that appear in $S$. For a length $m$ vector $y$ and a
subset $S \subset [m]$ we denote with $y|_S$ the vector that has exactly $|S|$ coordinates, which is obtained by
``throwing away'' all coordinates of $y$
outside $S$.

We denote $[k]=\{1,\ldots,k\}$. for a binary vector $w\in B^n$ we denote by $\wt(w)$, the weight of $w$, the number of nonzero entries of $w$.

%%%%%%%%%%%%%%%%%%%%%%%%%%%%%%%%%%%%%%
\section{Capacity achieving $t$-write  WOM codes}\label{sec:construction}
%%%%%%%%%%%%%%%%%%%%%%%%%%%%%%%%%%%%%%

\subsection{Our approach}\label{sec:method}

We describe our technique for proving Theorem~\ref{thm:main}. We start by explaining the connection to the $2$-write WOM codes of \cite{Shpilka-WOM12}, which themselves follow ideas from \cite{Wu10,YKSVW10}.

Similarly to \cite{RivestShamir82}, in the first round
\cite{Wu10,YKSVW10,Shpilka-WOM12} think of a message as a subset $S\subset [n]$ of
size $pn$ and encode it by its characteristic vector. This is also what we do in this paper.
%It is not hard to design efficient ways of encoding sets of size $pn$ as a
%number between $0$ and ${n\choose pn}-1$.
Clearly in this step we can transmit $H(p)n$ bits of information.
(I.e. $\log|\Omega_1| \approx H(p)n$.)

For the second round assume that we already sent a message
$S\subset [n]$. I.e. we have already wrote  $pn$ locations. Note
that in order to match the capacity we should find a way to
optimally use the remaining $(1-p)n$ locations in order to
transmit $(1-p-o(1))n$ many bits. \cite{Shpilka-WOM12} handled this by giving a collection of binary codes that
are, in some sense, {\em MDS codes on average}. Namely, a
collection of (less than) $2^{n}$ matrices $\{A_i\}$ of size
$(1-p-\epsilon)n\times n$ such that for any subset $S\subset[n]$,
of size $pn$, all but a fraction $2^{-\epsilon n}$ of the matrices
$A_i$, satisfy that $A_i|_{[n]\setminus S}$ has full row rank
(i.e. rank $(1-p-\epsilon)n$). Now, if in the first round
we transmitted a word $w$ corresponding to a subset $S\subset[n]$
of size $pn$ then in the second round we find a matrix $A_i$ such that
$A_i|_{[n]\setminus S}$ has full row rank. Now, given a word $x$ of length $(1-p)n$ to write in the second round, we find (by solving a system of linear equations) a vector $y\geq w$ such that $A_i y=x$. We then write $y$ to memory. Note, however, that we also have to write the index $i$ to memory so the decoder knows how to decode $y$ (the decoding function simply multiplies $y$ by $A_i$). However, since $i$ ranges over a set of size $2^n$, its index has length $n$ so writing it to memory reduces the rate significantly. Furthermore, to find the good $A_i$ we need to range over an exponentially large space, namely, a space of size $2^n$. Nevertheless, \cite{Shpilka-WOM12} showed that we can use the same $A_i$ for encoding many messages $x_i$ in parallel, so by concatenating $N=\exp(n)$ many copies of this basic construction, we obtain a code with the required properties. Indeed, now, a running time of $\exp(n)$ is merely polynomial in $N$.

The approach that we apply here is similar but with a few notable changes. Since we wish to have more than two write rounds we cannot use all the $(1-p)n$ bits left after writing the subset $S$ to memory. Instead we have to leave many unwritten bits for the third round and then to the fourth round etc. To do so, and still follow the previous approach, we would like to find a {\em sparse} vector $y\geq w$ that still satisfies $A_i y=x$, for some $A_i$ in the collection. For certain, to ensure that we can find such a $y$ we must know that such $y$ indeed exists. Namely, that there is some $A_i$ that has such a sparse solution above $w$ to the equality $A_i y=x$. Moreover, if we wish to continue with the analogy then we need to find such $A_i$ that will be good simultaneously to many equations of the form $A_i y_j=x_j$ where $y_j\geq w_j$. Indeed this is what we do. While the work \cite{Shpilka-WOM12} used the Wozencraft ensemble of matrices, in this work we use a universal family of pairwise independent hash-functions. (We describe the family and its properties in Section~\ref{sec: hash}.) We show that in this family we are always guaranteed to find a matrix for which there exist sparse solutions (that sit above the relevant $w_i$'s) to several such equations. This basically follows from the fact that this family is an extractor with an exponentially small error. (We state the required properties, which are weaker than what can actually be proved, in Lemma~\ref{lemma: hash is good}.) However, it is still not clear how to find the solutions $y_j$ even if they exist. The point is that by solving a system of linear equations we are not guaranteed to find a sparse solution. Here we again use the fact that we are allowed to concatenate many copies of our construction in parallel. Indeed, as finding a good $A_i$ already requires exponential time, we can let ourselves look for a solution $y_j$ in exponential time without hurting the already exponential running time. At the end, when we concatenate exponentially many copies of this basic construction, it will look as if we are solving in a brute force manner a problem on a logarithmically small space, which at the end is not too terrible. The effect of concatenating exponentially many copies of the basic construction, besides making the running time a small polynomial in the block length, is that the output block length is exponential in $1/\epsilon$, where $\epsilon$ is the gap to capacity that we achieve. 

We note that the construction of \cite{Shpilka-WOM12} suffered a similar blow-up in block length, from essentially the same reasons. The same is also true for the best codes of  \cite{YaakobiShpilka12} as their construction relies on  some of the tools developed in \cite{Shpilka-WOM12}.

\subsection{Hash functions}\label{sec: hash}

Our construction will make use of the following family of pairwise independent hash functions.
It will be instructive to think of the field $\F_{2^n}$ both as a field and as a vector space over $\F_2$. Thus, each vector $x\in\F_{2^n}$ can be naturally described as a length $n$ vector, after we fix some basis to $\F_{2^n}/\F_2$. For $a,b\in \F_{2^n}$ we define the map $H^{n,k,\ell}_{a,b}:\B^n\to\B^{k-\ell}$ as $H^{n,k,\ell}_{a,b}(x) = (ax+b)_{[k-\ell]}$. In words, we compute the affine transformation $ax+b$ in $\F_{2^n}$, represent it as an $n$-bit vector using the natural map and then keep the first $k-\ell$ bits of this vector. Note that we think of $x$ both as an $n$-bit vector and as an element of $\F_{2^n}$. We represent this family by ${\cal H}^{n,k,\ell}$. Namely
\begin{equation*}
{\cal H}^{n,k,\ell} = \left\{ H_{a,b}^{n,k,\ell} \mid a,b \in \F_{2^n}\right\}.
\end{equation*}

We will use the following special case of the leftover hash lemma of \cite{ImpagliazzoLL89}.

\begin{lemma}[Leftover hash lemma \cite{ImpagliazzoLL89}]\label{lemma: hash is good}
Let $k,\ell,n$ be integers. Let $Y\subseteq\B^{n}$ be a set of size $2^{k}$. Then, the distribution $(a,b,H^{n,k,\ell}_{a,b}(y))$, obtained by picking $a,b\in\F_{2^n}$ and $y\in Y$ independently and uniformly at random, is $2^{-\ell/2}$ close, in statistical distance,\footnote{The statistical distance between two distribution $D_1$ and $D_2$ defined over a set $X$ is $\max_{T\subseteq X} |\Pr_{D_1}[T]-\Pr_{D_2}[T]|$. Note that this is half the $\ell_1$ distance between $D_1$ and $D_2$.} to the uniform distribution on $\B^{2n+k-\ell}$.
\end{lemma}

\begin{corollary}\label{cor: has support}
Let $Y\subseteq\B^{n}$ be a set of size $2^{k}$. The fraction of $H\in {\cal H}^{n,k,\ell}$ such that $|H(Y)|\leq 2^{k-\ell}(1 - 2^{-\ell/4})$ is at most $2^{-\ell/4}$.
\end{corollary}

\begin{proof}
Observe that if $|H(Y)|\leq 2^{k-\ell}(1 - 2^{-\ell/4})$ then the statistical distance between $H(Y)$ and the uniform distribution on $\B^{k-\ell}$ is at least $2^{-\ell/4}$. If there is a fraction $>2^{-\ell/4}$ of such ``bad'' $H$ then the statistical distance between the uniform distributions on $(H,H(Y))$ and ${\cal H}^{n,k,\ell}\times \B^{k-\ell}$ is larger than $2^{-\ell/4}\cdot 2^{-\ell/4}=2^{-\ell/2}$, in contradiction to Lemma~\ref{lemma: hash is good}.
\end{proof}

The following corollary will be very useful for us. 

\begin{corollary}\label{cor:common solution}
Let $\ell,k,n,m$ be integers such that $\ell\leq k\leq n$ and $m< 2^{\ell/4}$. Let $Y_1,\ldots,Y_m \subseteq\B^n$ be sets of size $|Y_1|,\ldots,|Y_m|\geq 2^k$. Then, for any $x_1,\ldots,x_m\in\B^{k-\ell}$ there exists $H\in {\cal H}^{n,k,\ell}$ and $\{y_i\in Y_i\}$ such that for all $i$, $H(y_i)=x_i$. 
\end{corollary}

In words, the corollary says that if $m< 2^{\ell/4}$ then we can find an $H\in {\cal H}^{n,k,\ell}$ such that for all $i$, $x_i$ is in the image of $H$ when applied to $Y_i$ (i.e. $x_i\in H(Y_i)$). 

\begin{proof}
By Corollary~\ref{cor: has support} and the union bound we get that there is some $H^{n,k,\ell}_{a,b}\in{\cal H}^{n,k,\ell}$ such that for all $i$, $|H^{n,k,\ell}_{a,b}(Y_i)|> 2^{k-\ell}(1 - 2^{-\ell/4})$. Denote $X_i=H^{n,k,\ell}_{a,b}(Y_i)$. Observe, that by the union bound, there exists $v\in \B^{k-\ell}$ such that for all $i$, $x_i\oplus v\in X_i$. Thus, there exist $\{y_i\in Y_i\}$ such that for all $i$, $H^{n,k,\ell}_{a,b}(y_i)=x_i\oplus v$. In other words, for all $i$, $H^{n,k,\ell}_{a,b\oplus v\circ \vec{0}}(y_i)=x_i$, where by $v\circ \vec{0}$ we mean a length $n$ vector that is equal to $v$ on the first $k-\ell$ bits and is zero elsewhere. $H=H^{n,k,\ell}_{a,b\oplus v\circ\vec{0}}$ is a function satisfying the claim.
\end{proof}

\subsection{The basic construction}\label{sec:basic}

We now describe a basic construction that will be the building block of our final encoding scheme. The final construction will be composed of concatenating several blocks of the basic construction together. To make reading easier, we drop all notation of floors and ceilings, but the reader should bear in mind that not all numbers that we may encounter in the analysis are integers and so some rounding takes place, but we choose parameters in such a way that we absorb all loss created by such rounding. 

\begin{lemma}\label{lem: analysis basic block}
For any $\epsilon>0$, $(R_1,\ldots,R_t)\in C_t$ and $1\leq t_0 \leq t$ the following holds. 
Let $n$ be an integer such that $n = c\log(1/\epsilon)/\epsilon$, for some constant $c>20$. Let $k,\ell,m$ be integers satisfying $k= (R_{t_0} - \epsilon/3) n$, $\ell = \epsilon n/3=c\log(1/\epsilon)/3$ and $m<2^{\ell/4}=(1/\epsilon)^{c/12}$. 
 Let $w_1,\ldots,w_m\in\B^n$ be such that 
for all $i$: if $t_0=1$ then $\wt(w_i)=0$, otherwise, 
%$$\wt(w_i)\leq \left(p_1 + \sum_{s=1}^{t_0-2} \left(\prod_{j=1}^{s}(1-p_j)\right) \cdot p_{s+1}\right)n.$$
$\wt(w_i)\leq \left(1 - \prod_{j=1}^{t_0-1}(1-p_j)\right)n.$
Let $x_1,\ldots,x_m\in \B^{k-\ell}$. Then, there exists $H_{t_0}\in {\cal H}^{n,k,\ell}$ and $y_1,\ldots,y_m\in \B^n$ with the following properties:
\begin{enumerate}
\item For all $i\in [m]$, $y_i\geq w_i$.
\item For all $i\in [m]$, 
%$$\wt(y_i)\leq \left(p_1 + \sum_{s=1}^{t_0-2} \left(\prod_{j=1}^{s}(1-p_j)\right) \cdot p_{s+1} + \left(\prod_{j=1}^{t_0-1}(1-p_j)\right) \cdot p_{t_0}\right) n.$$
$\wt(y_i)\leq \left(1 - \prod_{j=1}^{t_0}(1-p_j)\right)n.$
\item For all $i\in [m]$, $H_{t_0}(y_i)=x_i$.
\end{enumerate}
Furthermore, we can find such $H_{t_0}$ and $y_1,\ldots,y_m$ in time $2^{4n}$.
\end{lemma}

Intuitively, one should think of each $w_i$ in the statement of the lemma as the state of the memory after $t_0-1$ rounds of writing (i.e. before round $t_0$). The map $H_{t_0}$ that the lemma guarantees will be the encoding map used in round $t_0$ and each $y_i$ will be the word replacing $w_i$ in memory. (We are working in parallel with several $w_i$ and $y_i$ but with a single $H_{t_0}$.)

\begin{proof}
Let
\begin{equation*}
Y_i = \left\{y\in \B^n \mid y\geq w_i \text{ and } \wt(y)\leq \left(1 - \prod_{j=1}^{t_0}(1-p_j)\right)n \right\}.
\end{equation*}
As $\wt(w_i)\leq \left(1 - \prod_{j=1}^{t_0-1}(1-p_j)\right)n$ we have that 
$$|Y_i|\geq {n-\wt(w_i) \choose  \prod_{j=1}^{t_0-1}(1-p_j)\cdot p_{t_0}\cdot n}\geq { \prod_{j=1}^{t_0-1}(1-p_j)n  \choose  \prod_{j=1}^{t_0-1}(1-p_j)\cdot p_{t_0}\cdot n }$$
which by Stirling's formula can be lower bounded by 
\begin{eqnarray*}
\geq 2^{\prod_{j=1}^{t_0-1}(1-p_j) \cdot H(p_{t_0})\cdot n - \log (\prod_{j=1}^{t_0-1}(1-p_j)n)}&\geq & 2^{\prod_{j=1}^{t_0-1}(1-p_j) \cdot H(p_{t_0})\cdot n - \log (n)}\\
 &> & 2^{ (R_{t_0} - \epsilon /3)n}=2^k.
\end{eqnarray*}
We can now apply Corollary~\ref{cor:common solution} to find $H_{t_0}\in{\cal H}^{n,k,\ell}$ and $\{y_i\in Y_i\}$ such that for all $i$, $H_{t_0}(y_i)=x_i$. As $y_i\in Y_i$ we see that all the three requirements in the lemma hold.  

As for the furthermore part, we note that since we are guaranteed that such $H$ and $\{y_i\}$ exist, we can find them in a brute-force manner in time $|{\cal H}^{n,k,\ell}|\cdot \sum_{i=1}^{m}|Y_i|\leq m2^{3n}<2^{4n}$. 
\end{proof}

Note that we need $n$ to satisfy $\epsilon n > 3\log(n)+9$. Thus, $n\geq 20\log(1/\epsilon)/\epsilon$ suffices. 

The lemma basically tells us that if we have $m$ memory blocks, each of length $n$, then no matter what we wrote to each memory block in the first $t_0-1$ rounds, we can find suitable $\{y_i\}$ for writing in the next round. Notice that in order to decode $x_i$ one simply needs to compute $H_{t_0}(y_i)$. However, we somehow need to communicate $H_{t_0}$ to the decoder as it may depend on the content of the memory after the first $t_0-1$ rounds. In order to do so we write the $2n$-bit index of $H_{t_0}$ on fresh memory cells. (Recall that each $H\in {\cal H}^{n,k,\ell}$ is defined by a pair $(a,b)\in \F_{2^n}\times \F_{2^n}$.) Wrapping everything up we now define the basic building block of our $t$-write WOM scheme.

\paragraph{Basic WOM scheme:} For any given $\epsilon>0$ and $(R_1,\ldots,R_t)\in C_t$ we define the following encoding scheme. Let $c>20$ be an integer satisfying $6t^2<(t/\epsilon)^{c/12 -1}$.
Let $n=ct\log(t/\epsilon)/\epsilon$ be an integer. Let $k_2,\ldots,k_t,\ell,m$ be integers satisfying $k_i= (R_{i} - \epsilon/3t) n$, $\ell = \epsilon n/3t=c\log(t/\epsilon)/3$ and $m=2^{\ell/4}-1=(t/\epsilon)^{c/12}-1$. Let $N_0=2n(t-1)+mn$. Our encoding scheme will write to memory of length $N_0$. We shall think of the $N_0$ memory bits as being composed of $m$ blocks of $n$ bits where we write the ``actual'' words and additional $t-1$ blocks of $2n$ bits each where we store information that will help us decode the stored words.
\begin{itemize}
\item {\bf First round:} A message in the first round is given as a list of $m$ subsets of $[n]$, each of size at most $p_1 n$. We store the $i$th subset to the $i$th $n$-bit block by simply writing $1$ on all indices belonging to the set and $0$ elsewhere. It is clear how to decode the message written in the first round. 
\item {\bf Round $j$:} The input message for round $j$ is composed of $m$ words $x_1,\ldots,x_m \in \B^{k_j-\ell}$. Assume that the $m$ $n$-bit memory blocks contain the words $w_1,\ldots,w_m$. We first find $H_j$ and $y_1,\ldots,y_m$ as guaranteed by Lemma~\ref{lem: analysis basic block}. We write $y_i$ to the $i$th $n$-bit memory block and we also write the index of $H_j$ to the $j-1$ block of length $2n$. To decode we simply read the index of $H_j$ from the $j-1$th $2n$-bit block, and apply $H_j$ on each $y_i$ to get $x_i$. 
\end{itemize}

\paragraph{Analysis:} 
\begin{itemize}
\item {\bf Correctness:} By induction one can easily show that if $w$ was written to an $n$-bit memory block at round $t_0\geq 1$ then $\wt(w)\leq \left(1 - \prod_{j=1}^{t_0}(1-p_j)\right)n$. Thus, at the $j$th round the conditions of Lemma~\ref{lem: analysis basic block} are satisfied (assuming that we pick $n$ large enough) and therefore we can find $H_j$ and $\{y_i\}$ as required. 

\item {\bf Rate:} We first note that by our choice of parameters $$\frac{mn}{N_0} = \frac{mn}{(mn+2n(t-1))}=1-\frac{2(t-1)}{m+2(t-1)}>1-\frac{2t}{(t/\epsilon)^{c/12}}>1-\epsilon/3t.$$ In the first round we write a word of weight $p_1 n$. Thus the rate of this round is $$\frac{m\log\left({n \choose p_1 n}\right)}{N_0} > \frac{\left(H(p_1)n - \log(n)\right)m}{N_0} > \frac{\left(H(p_1)-\frac{\epsilon}{3t}\right)nm}{N_0}\geq (R_1-\frac{\epsilon}{3t})(1-\frac{\epsilon}{3t})>R_1-\frac{\epsilon}{t}.$$  In the $j$th round we encode $m$ messages each of length $k_j-\ell$. Thus, the total rate is $$\frac{m(k_j-\ell)}{N_0}\geq \frac{m((R_j-\epsilon/3t)-\epsilon/3t)n}{N_0}>(R_j-\frac{2\epsilon}{3t})(1-\frac{\epsilon}{3t})>R_j-\frac{\epsilon}{t}.$$ It follows that the total rate is at least $R_1+\ldots+R_t - t\frac{\epsilon}{t} = R_1+\ldots +R_t - \epsilon$.

\item {\bf Encoding Complexity:} by Lemma~\ref{lem: analysis basic block} each round can be completed in time $2^{4n}=(t/\epsilon)^{4ct/\epsilon}$. 

\item {\bf Decoding complexity:} To decode a message written in the $j$th round (assume $j>1$ as the case $j=1$ is clear) we simply need to perform a product of an $(k_j-\ell)\times n$ matrix with $m$ $n$-bit vectors (we can think of each $H_{a,b}$ as an affine transformation and associate to it a $(k-\ell)\times n$ matrix and a $(k_j-\ell)$ shift vector that completely define it). Thus, the decoding running time  is $\poly(kmn)=(t/\epsilon)^{O(c)}$.\footnote{Here we assume that the field $\F_{2^n}$ is explicitly given by an irreducible polynomial, so given $a\in\F_{2^n}$ we can find the $n\times n$ matrix $A$ over $\F_2$ satisfying $Ax=ax$ for all $x\in \F_{2^n}$ (when we abuse notation and think of $x$ both as an $n$-bit vector and field element) in time $\poly(n)$.}
\footnote{We do not try to optimize running time here.}
\end{itemize}

We summarize our conclusions in the following theorem.

\begin{theorem}\label{thm:main basic}
For any given $\epsilon>0$ and $(R_1,\ldots,R_t)\in C_t$ let $c>20$ an integer satisfying $6t^2<(t/\epsilon)^{c/12 -1}$.
Let $n=ct\log(t/\epsilon)/\epsilon$ and $m=(t/\epsilon)^{c/12}-1$ be integers. Let $N_0=2n(t-1)+mn$. Then, the scheme described above is a $t$-write WOM scheme of block length $N_0$, with rate larger than $R_1+\ldots +R_t - \epsilon$, encoding complexity of  $(t/\epsilon)^{4ct/\epsilon}$ and decoding complexity $(t/\epsilon)^{O(c)}$.
\end{theorem}

\subsection{The full construction}\label{sec:full}

The full construction is obtained by concatenating (in the sense of string concatenation) several blocks of the basic construction together. 

\begin{theorem}\label{thm: full construction}
Given $a,\epsilon>0$ and $(R_1,\ldots,R_t)\in C_t$ let $c>20$ an integer satisfying $6t^2<(t/\epsilon)^{c/12 -1}$.
Let $n=ct\log(t/\epsilon)/\epsilon$, $m=(t/\epsilon)^{c/12}-1$ and $N_0=2n(t-1)+mn$ be integers. Set $N_1=2^{\frac{4n}{a-1}}\cdot N_0$. Then, there is an explicit $t$-write WOM scheme of block length $N_1$, with rate larger than $R_1+\ldots +R_t - \epsilon$, encoding complexity of  $N_1^{a}$ and decoding complexity $N_1^{1+O(a\epsilon/t)}$.
\end{theorem}

\begin{proof}
Let $n_1=(N_1/N_0)=2^{4n/(a-1)}$. Our scheme is composed of concatenating $n_1$ copies of the basic construction. Thus, the block length that we obtain is $n_1\cdot N_0=N_1$. It is clear that concatenation does not change the rate and the effect on the complexity is that we need to repeat all operations $n_1$ times. The claimed running times now follow from Theorem~\ref{thm:main basic}.
\end{proof}

Theorem~\ref{thm:main} now follows by appropriately choosing $a=1+\delta$. Note that 
$$N_1 = n_1\cdot N_0 = O(n_1\cdot n\cdot m)=O\left( \left(\frac{t}{\epsilon}\right)^{\frac{4ct}{\epsilon(a-1)}}\cdot \left(\frac{t}{\epsilon}\right)^{c/12}\cdot \frac{ct\log(t/\epsilon)}{\epsilon}\right)=\left(\frac{t}{\epsilon}\right)^{O\left(\frac{ct}{\epsilon\delta}\right)}.$$

\subsection{Larger alphabets}

We only give a sketch of the proof of Theorem~\ref{thm:main-large} as it is nearly identical to the proof of Theorem~\ref{thm:main}. 

As before we shall think of the field $\F_{q^n}$ both as a field and as an $n$-dimensional vector space over $\F_q$. 
For $a,b\in \F_{q^n}$ we define the map $H^{n,k,\ell}_{a,b}:\F_q^n\to\F_q^{k-\ell}$ as $H^{n,k,\ell}_{a,b}(x) = (ax+b)_{[k-\ell]}$. We also set
\begin{equation*}
{\cal H}_q^{n,k,\ell} = \left\{ H_{a,b}^{n,k,\ell} \mid a,b \in \F_{q^n}\right\}.
\end{equation*}
The first step is to note the following analogs to Lemma~\ref{lemma: hash is good} and Corollary~\ref{cor:common solution}.

\begin{lemma}[Leftover hash lemma over $\F_q$ \cite{ImpagliazzoLL89}]\label{lemma: large hash is good}
Let $k,\ell,n$ be integers. Let $Y\subseteq\F_q^{n}$ be a set of size $q^{k}$. Then, the distribution $(a,b,H^{n,k,\ell}_{a,b}(y))$, obtained by picking $a,b\in\F_{q^n}$ and $y\in Y$ independently and uniformly at random, is $q^{-\ell/2}$ close, in statistical distance, to the uniform distribution on $\F_q^{2n+k-\ell}$.
\end{lemma}

\begin{corollary}\label{cor:large common solution}
Let $q,\ell,k,n,m$ be integers such that $\ell\leq k\leq n$ and $m< q^{\ell/4}$. Let $Y_1,\ldots,Y_m \subseteq\F_q^n$ be sets of size $|Y_1|,\ldots,|Y_m|\geq q^k$. Then, for any $x_1,\ldots,x_m\in\F_q^{k-\ell}$ there exists $H\in {\cal H}_q^{n,k,\ell}$ and $\{y_i\in Y_i\}$ such that for all $i$, $H(y_i)=x_i$. 
\end{corollary}

Given Corollary~\ref{cor:large common solution}, we can construct, in a similar manner to Sections~\ref{sec:basic} and \ref{sec:full}, schemes for $t$-write $q$-ary WOM codes. The only difference is that the matrix $S_i$ (recall discussion before the statement of Theorem~\ref{thm:main-large}) will be used to define the sources $Y$ in which we will be looking for solutions in the $i$th write. Namely, give a word $w$ in memory, the relevant codewords for the $i$th write will be those that have a fraction $(S_i)_{j_2,j_1}$ coordinates labeled $j_2$ among all coordinates where $w$ had value $j_1$, namely,  we shall work with the sources $$Y_i(w) = \left\{w' \mid \#\{ b \mid w'_b = j_2 \text{ and } w_b=j_1\} \approx (S_i)_{j_2,j_1} \cdot \#\{ b \mid w_b=j_1\}\right\} .$$

\section{Discussion}

In this work we showed how to construct capacity achieving $t$-write WOM codes.
The main drawback of our codes is that the block length is exponential in $1/\epsilon$, where $\epsilon$ is the gap to capacity. It is an interesting question, that may have practical significance, to construct capacity achieving codes with a much shorter block length $N=\poly(1/\epsilon)$. If we would like to get such a result using  similar techniques, then, in particular, we will need a way to efficiently find sparse solutions to systems of linear equations of restricted forms. We know that finding sparse solutions to general linear systems is NP-hard, but as the matrices that we used for our construction are very structured it does not sound improbable that such an algorithm should exist. Moreover, one can also imagine a different set of matrices that can be used for decoding for which sparse solutions can be found efficiently.  

Another way to think of our use of the universal hash family is that what we actually used is a very good extractor (for background on extractors we refer the reader to \cite{Shaltiel02}). Namely, an extractor with an exponentially small error. We note that we actually do not need such good extractors. We merely gave this construction as it is easier to explain. In particular,  one can use \cite{Rao09,Shaltiel08} to replace the ensemble of matrices with a single function (i.e. with a seedless extractor for several sources). While this will save a bit on the block length (though will not affect it in a significant manner) the main hurdle which is efficiently {\em inverting} the extractor still remains. Thus, it seems that the main challenge, as far as our approach goes, is to efficiently construct seedless extractors (or extractors with an exponentially small error) that can be efficiently inverted, assuming that the source has some known ``nice'' structure. In our case all the sources that we deal with are highly structured - they contain all vectors of a certain weight that ``sit above'' a given vector. In a sense, this is the ``positive'' part of the hamming ball centered at the given vector. We think that this is an interesting challenge and we expect that developments regarding this question will lead to improved constructions of WOM codes. \\

\cite{Shpilka-WOM12} also considered the related problem of designing encoding schemes for defective memory and gave a scheme based on the extractor of Raz, Reingold and Vadhan \cite{RazRV02}. The disadvantage is that for the scheme to work \cite{Shpilka-WOM12} needed to assume that there are $\poly(\log n)$ ``clean bits'' that are known to both the encoder and decoder. Recently Gabizon and Shaltiel considered the problem and showed how to obtain capacity achieving encoding schemes from zero-error seedless dispersers for affine sources \cite{GabizonShaltiel12}.\footnote{Roughly, a zero-error disperser for affine sources is a function $f\B^n \to \B^m$ such that any affine space $V$, of high enough dimension, is mapped onto $\B^m$. That is, for any such $V$, $f(V)=\B^m$.} Similarly to our work, \cite{GabizonShaltiel12} needed their disperser to be efficiently invertible. Since they were dealing with affine sources (and not with Hamming balls as we do here) they managed to achieve this by cleverly combining a search over a small space with solving a system of linear equations.   We also note that the encoding function of \cite{GabizonShaltiel12} is randomized (with expected running time $\poly(n)$), whereas the scheme of \cite{Shpilka-WOM12}, and the one that we give in this paper, are deterministic. Nevertheless, it may be possible that using techniques similar to those of \cite{GabizonShaltiel12} one may be able to improve on the block-length of our construction.

\section*{Acknowledgments}

Most of this work was done while the author was visiting at the school of computer and communication sciences at EPFL. We wish to thank Rudiger Urbanke for his hospitality and for fruitful discussions. We thank David Burshtein for answering our questions regarding  \cite{BurshteinS12} and we also thank Ariel Gabizon and Ronen Shaltiel for their comments.
We are grateful to Eitan Yaakobi for commenting on an earlier draft of this paper and for many discussions on WOM codes.

\bibliographystyle{alpha}

%\bibliography{../../bibliography}

\newcommand{\etalchar}[1]{$^{#1}$}

\end{document}